\documentclass[english,prl,twocolumn,superscriptaddress]{revtex4-1}
\usepackage[english]{babel}
\usepackage[utf8]{inputenc}
\usepackage{amsmath}
\usepackage{amssymb}
\usepackage{amsthm}
\usepackage{graphicx}
\usepackage{verbatim}

\def\C{\mathbb C}
\def\Q{\mathbb Q}
\def\K{\mathbb K}
\def\CP{\mathbb{CP}}
\def\mc{\mathcal}
\def\nn{\nonumber}
\def\eps{\epsilon}
\def\avg#1{\langle#1\rangle}
\def\Res{\operatorname{Res}}
\def\CHY{{\operatorname{CHY}}}
\def\vol{\operatorname{vol}\,\!}
\def\SL{\operatorname{SL}}
\def\LT{\operatorname{LT}}

\def\tree{{\operatorname{tree}}}
\def\In{\operatorname{in}}

\newtheorem{thm}{Theorem}

\newtheorem{corollary}{Corollary}[thm]
\newtheorem{definition}[thm]{Definition}
\newtheorem{proposition}[thm]{Proposition}

\begin{document}
\date{\today}
\author{Jorrit Bosma}
\affiliation{
Institute for Theoretical Physics, ETH Z{\"u}rich, \\
Wolfgang-Pauli-Strasse 27, CH-8093 Z{\"u}rich, Switzerland
}
\author{Mads S{\o}gaard}
\affiliation{
SLAC National Accelerator Laboratory, Stanford University, \\
2575 Sand Hill Road, Menlo Park, CA 94025, USA
}
\author{Yang Zhang}
\affiliation{
Institute for Theoretical Physics, ETH Z{\"u}rich, \\
Wolfgang-Pauli-Strasse 27, CH-8093 Z{\"u}rich, Switzerland
}

\title{The Polynomial Form of the Scattering Equations is an H-Basis}

\begin{abstract}
We prove that the polynomial form of the scattering equations is a Macaulay
H-basis. We demonstrate that this H-basis facilitates integrand reduction and
global residue computations in a way very similar to using a Gr{\"o}bner basis,
but circumvents the heavy computation of the latter. As an example, we apply the
H-basis to prove the conjecture that the dual basis of the polynomial scattering
equations must contain one constant term.
\end{abstract}

\maketitle

\section{Introduction}
The Cachazo, He and Yuan (CHY) formulation 
\cite{Cachazo:2013gna,Cachazo:2013hca,Cachazo:2013iea,Cachazo:2014nsa,Cachazo:2014xea}
of the perturbative $S$-matrix relies on a collection of rational maps $\{f_a\}$
from the space of massless kinematic configurations to the moduli space 
$\mc M_{0,n}$ of Riemann spheres with $n$ marked points $z_a\in\CP^1$ associated
with each of the external particles. Particularly, the set $\mc S$ of the
$(n-3)!$ solutions to simultaneous constraints,
\begin{align}
f_a(z,k) = \sum_{b\neq a}\frac{k_a\cdot k_b}{z_a-z_b} = 0\;, \quad
\forall a\in A = \{1,2,\dots,n\}\;.
\label{EQ_SCATTERING_EQS}
\end{align}
dubbed as the {\it scattering equations}, provides a very intriguing basis for
decomposing massless scattering processes in generic quantum field theories. The
external particles have momenta $k_a$ and polarizations $\eps_a$. Scattering
amplitudes in arbitrary dimension are expressed as a multidimensional integral
of a certain rational function on $\mc M_{0,n}$
\cite{Cachazo:2013gna,Cachazo:2013hca,Cachazo:2013iea,Cachazo:2014nsa,Cachazo:2014xea}.
At the heart of the formalism lies the principle that the integration is
localized on the support of the scattering equations. We write amplitudes as
\begin{align}
\mc A_n^\tree(\{k_a,\eps_a\}) = 
\int d\Omega_\CHY\mc I(\{z_a\},\{k_a,\eps_a\})\;,
\label{EQ_CHY_FORMULA}
\end{align}
where $d\Omega_\CHY$ is the integration measure,
\begin{align}
d\Omega_\CHY\equiv
\frac{d^nz}{\vol\SL(2,\C)}
{\prod_a}'\delta(f_a)\;,
\label{EQ_CHY_MEASURE}
\end{align}
and $\mc I$ is referred to as the CHY integrand. The latter is a rational
function of the marked points $z_a$ and $d\Omega_\CHY$ is constructed from the
$f_a$'s. The philosophy is that for a given theory in consideration, for example
Yang-Mills theory, there exists is a compact integrand $\mc I$ such that
eq.~\eqref{EQ_CHY_FORMULA} reproduces the correct $S$-matrix. 

In order to briefly explain the notation in eq.~\eqref{EQ_CHY_MEASURE}, we
remark that $\SL(2,\C)$ invariance implies that imposition of merely any $(n-3)$
of the scattering equations suffices to restrict the solution to
eq.~\eqref{EQ_SCATTERING_EQS}. Evidently,
\begin{align}
{\prod_a}'\delta(f_a)\equiv
z_{ij}z_{jk}z_{ki}
\prod_{a\in A\backslash\{i,j,k\}}\delta(f_a)\;,
\end{align}
where the labels $i,j,k$ specify the arbitrary choice of the the three
extraneous scattering equations to be disregarded. Throughout this paper
$z_{ab}\equiv z_a-z_b$. Moreover, the $\SL(2,\C)$ redundancy is explicitly
quotiented out by fixing the values of, say, $z_r$, $z_s$ and $z_t$.

For the purpose of investigating aspects of the CHY formalism via algebraic
geometry, it is essential to interpret eq.~\eqref{EQ_CHY_FORMULA} as a
multivariate {\it global residue} with respect to the polynomial form of the
scattering equations derived by Dolan and Goddard \cite{Dolan:2014ega}. Let
$h_m$ be the multilinear homogeneous polynomial of degree $m$ defined by
\begin{align}
h_m = \frac{1}{m!}\sum_{
\substack{
a_1,a_2,\dots,a_m\in A' \\
a_i\neq a_j}}
\sigma_{a_1a_2\cdots a_m}z_{a_1}z_{a_2}\cdots z_{a_m}\;,
\label{EQ_DOLAN_GODDARD}
\end{align}
for $A' = \{2,\dots,n-1\}$ and $\sigma_{a_1a_2\cdots a_m}\equiv
k_{1a_1a_2\cdots a_m}^2$. Then eqs.~\eqref{EQ_SCATTERING_EQS} subject to the
partial gauge fixing $z_1\to 0$ and $z_n\to\infty$ are equivalent to the
polynomial equations,
\begin{align}
h_m = 0\;, \quad
1\leq m \leq n-3\;.
\label{EQ_CHY_POLY}
\end{align}
In terms of the $h_m$'s, eq.~\eqref{EQ_CHY_FORMULA} can be rewritten as
$(n-3)$-fold integral over a contour $\mc O$ encircling all points in $\mc S$,
with the replacements $\mc I\to\tilde{\mc I}$ and $d\Omega_\CHY\to
d\tilde{\Omega}_\CHY$,
\begin{align}
\tilde{\mc I}\equiv {} & 
\mc I\prod_{a\in A}(z_a-z_{a+1})^2\;, \\[-1mm]
d\tilde{\Omega}_\CHY\equiv {} & 
\frac{z_2}{z_{n-1}}
\prod_{m=1}^{n-3}\frac{1}{h_m(z,k)} \\[-1mm] & \times
\prod_{2\leq a<b\leq n-1}\!(z_a-z_b)
\prod_{a=2}^{n-2}\frac{z_adz_{a+1}}{(z_a-z_{a+1})^2}\;. \nn
\label{EQ_CHY_DG}
\end{align}

The CHY literature is by now fairly extensive. We suggest a partial list
\cite{Dolan:2014ega,Weinzierl:2014vwa,Dolan:2013isa,Weinzierl:2014ava,Dolan:2015iln,Baadsgaard:2015ifa,Naculich:2014naa,Naculich:2015zha,Bjerrum-Bohr:2014qwa,Gomez:2016bmv,Cardona:2016bpi,Du:2016blz}
of recent developments. The loop-level generalization is addressed in
ref.~\cite{Adamo:2013tsa,Geyer:2015bja,Geyer:2015jch,He:2015yua,Cachazo:2015aol,Casali:2014hfa}.
Although eqs.~\eqref{EQ_SCATTERING_EQS} look very simple, it is a formidable
task to solve them to actually compute amplitudes. This problem has received
considerable attention recently, and it is now clear that the explicit solutions
can be bypassed completely. We mention the integration rules
\cite{Baadsgaard:2015voa,Baadsgaard:2015hia,Bjerrum-Bohr:2016juj,Feng:2016nrf,Huang:2016zzb}
and various other approaches
\cite{Kalousios:2015fya,Cardona:2015ouc,Huang:2015yka,Cardona:2015eba}. See also
refs.~\cite{Lam:2015sqb,Lam:2016tlk} for related progress. In
ref.~\cite{Sogaard:2015dba} two of the present authors proposed an advantageous
alternative offered by the Bezoutian matrix method from computational algebraic
geometry. We will revisit this part later. 

The main result of this paper is that the polynomial scattering equations
\eqref{EQ_DOLAN_GODDARD} automatically form an $H$-basis for the
zero-dimensional ideal $I = \avg{h_1,\dots,h_{n-3}}$.  

\section{H and G(r{\"obner}) Bases}
We consider the ring $R = \K[z_1,\dots,z_n]$ of polynomials in $n$ variables
$z_1,\dots,z_n$ over a field $\K$. Typically, $\K = \C$ or $\K = \Q$. Here we
follow refs.~\cite{HBasisI,CDS}.

Let $P_d$ be the subset of all polynomials of degree $d$ or less, and $S_d$ be
the subset of homogeneous polynomials of degree $d$. We have the direct sum
decomposition,
\begin{equation}
  \label{eq:1}
  P_d=\bigoplus_{i=0}^{d} S_d\;.
\end{equation}

Consider an ideal $I$ generated by polynomials $f_1,\dots,f_k$,
$I=\langle f_1,\dots,f_k \rangle$. So for any $f\in I$, 
\begin{equation}
  \label{eq:3}
  f= \sum_{i=1}^k q_i f_i, \quad q\in R\;.
\end{equation}
In practice, it is much easier to carry out polynomial reduction if 
\begin{equation}
  \label{degree_condition}
  \max_{i=1}^k \deg (q_i f_i) = \deg f\;,
\end{equation}
since the quotients $q_i$'s degrees are under control. However, the condition
\eqref{degree_condition} in general may not be satisfied for any set of
generators for $I$. Hence F. Macaulay \cite{Macaulay} defined the $H$-basis of
an ideal:
\begin{definition}
\label{hbasis_def}
We say that $\{f_1,\dots,f_k\}\subset I$ is an H-basis of an ideal
$I\subseteq R$, if $\forall f\in I$, $\exists q_1,\dots,q_k\in R$ such
that $f = \sum_{i=1}^k q_i f_i$, and 
\begin{align}
\max_{i=1}^k\deg(q_if_i) = \deg f\;.
\end{align}
\end{definition}

By this definition, the condition of being an $H$-basis is equivalent to that
$P_d \cap I$ is generated as,
\begin{equation}
  \label{degree_generator}
  P_d \cap I = {\rm span}_{\mathbb K} \{a f_i\}\;,
  \quad \forall a \in R,\ \deg a \leq d-\deg f_i\;.
\end{equation}

For any polynomial $f\in R$, define the {\it initial form} of $f$, $\In(f)$,
as the homogeneous part of $f$, with the degree $\deg f$. The condition of being
an $H$-basis can be further reformulated as \cite{HBasisI},
\begin{equation}
  \label{eq:4}
  \langle \In(I) \rangle = \langle \In(f_1),\ldots, \In(f_k)\rangle\;,
\end{equation}
where $\In(I)$ is the collection of initial forms of all polynomials in $I$. 

When the number of generators equals the number of variables, there is a simple
way to check if the generator set is an $H$-basis \cite{HBasisI,CDS}:
\begin{thm}
$\mc H\equiv\{f_1,\dots,f_k\}$ is an $H$-basis for the ideal  
$I=\langle f_1,\dots,f_k\rangle$, provided that $(0,\dots,0)$ is the only
simultaneous zero of the initial forms $\In(f_1),\dots,\In(f_k)$.
\label{THM_HBASIS}
\end{thm}

With an $H$-basis, many problems in commutative algebra can be translated to
linear algebra problems by eq.~\eqref{degree_generator}. We remark that the
$H$-basis, in many aspects, resembles the Gr\"obner basis ($G$-basis). They both
rely on the order of monomials. An $H$-basis sorts monomials by the degree,
while a $G$-basis sorts monomials by a total monomial order $\succ$. The
definition of a $G$-basis can be rephrased in a similar form of
def.~\ref{hbasis_def}: $\{g_1,\dots,g_k\}\subset I$ is a $G$-basis of an ideal
$I\subseteq R$ with respect to $\succ$, if $\forall f\in I$, $\exists
q_1,\dots,q_k\in R$ such that $f = \sum_{i=1}^k q_ig_i$, and 
\begin{equation}
\max_{i=1}^k \LT(q_ig_i) = {}  \LT(f)\;,
\end{equation}
where $\LT$ stands for the highest monomial with respect to $\succ$, and
``$\max$'' and ``$=$'' are to be understood in the context of this monomial
order.

The $G$-basis concept can be considered as a refined version of the $H$-basis,
since roughly speaking, it converts commutative algebra problems to
one-dimensional linear algebra problems. On the other hand, in many cases, like
the high-point scattering equations, the computation to obtain a $G$-basis is
heavy. Furthermore, for symmetric ideals, the $G$-basis introduces an artificial
order between variables which explicitly breaks the symmetry. 

\section{The Polynomial Form of the Scattering Equations is an H-basis}
We are now ready to prove our principal result.
\begin{thm}
The polynomial scattering equations $h_m$, $1\leq m\leq n-3$, form an $H$-basis.
\label{THM_PSE_HBASIS}
\end{thm}
\begin{proof}
We show that $V(J)\equiv V(\langle \In(h_1),\dots,\In(h_{n-3}) \rangle)$, the
zero locus of the initial forms of the polynomial scattering equations, consists
of only the point $(0,\dots,0)$, after which the result follows from Theorem
\ref{THM_HBASIS}. Note that since our field of interest $\K = \C$ is
algebraically closed, $V(J)=V(\sqrt{J})$, so it suffices to consider the radical
ideal $\sqrt{J} \supset J$. From inspection of eq.~\eqref{EQ_DOLAN_GODDARD},
after applying the gauge fixing $z_2\to 1$, we have
\begin{align}
\In(h_m) = 
\frac{1}{m!}\sum_{ 
\substack{
a_1,a_2,\dots,a_m\in A'' \\
a_i\neq a_j}}
\sigma_{a_1a_2\cdots a_m}z_{a_1}z_{a_2}\cdots z_{a_m}\;,
\end{align}
where $A'' = \{3,\dots,n-1\}$. By iteratively considering
\begin{equation}
z_{3}\cdots z_{n-1-i}\,h_{n-3-i} \in \sqrt{J}
\end{equation}
for $i=0,\dots,n-4,$ we find that, after each step, $z_{3}^{2}\cdots
z_{n-1-i}^{2} \in \sqrt{J}$ and hence $z_{3}\cdots z_{n-1-i} \in \sqrt{J}$. As
the same works for permutations of the indices, we realize that 
$z_{3},\dots,z_{n-1}\in\sqrt{J}$ and thus 
\begin{equation}
V(\sqrt{J})=V(\langle z_3,\dots,z_{n-1} \rangle) = \{(0,\dots,0)\}\;.
\end{equation}
\end{proof}

\section{Integrand Reduction}
One of the crucial features of an $H$-basis is that we can perform polynomial
reduction towards it, in a similar way as the $G$-basis. Here we do not need the
reduction algorithm for a generic $H$-basis \cite{HBasisI}, but only the case of
our interest, i.e. the polynomial scattering equations as an $H$-basis. The goal
is to reduce an arbitrary polynomial to a polynomial with degree low enough.

Let $\{f_1,\ldots f_n\}$ be an $H$-basis in $n$ variables $z_1,\dots,z_n$ and
let $I=\langle f_1,\ldots f_n \rangle$ so $\In(I)$ is the ideal of the initial
forms. As $V(\In(I))=\{(0,\ldots,0)\}$, by Hilbert's Nullstellensatz
\cite{MR1417938}, there exists a positive integer $d^*$ such that 
\begin{equation}
  \label{eq:2}
  S_{d} \subset \In(I), \quad \forall d>d^*.
\end{equation}
Now we estimate $d^*$ for scattering equations:
\begin{thm}
For the $n$-point tree-level scattering equations, define 
$d^*=(n-3)(n-4)/2$. Then $S_{d} \subset \In(I)$, $\forall d>d^*$.
\label{CRITICAL_DEGREE}
\end{thm}
\begin{proof}
Homogeneity of the initial forms implies that 
$R/\langle \In(I)\rangle=\bigoplus_{i=0}^\infty A_d$ is a graded algebra. Then,
\begin{equation}
  \label{eq:5}
  \dim_{\mathbb K}A_d =0\;,\quad \text{if } d>\sum_{i=1}^{k} \deg f_i -n\;,
\end{equation}
when $V(\In(I))=\{(0,\ldots,0)\}$ \cite{MR2161985}. The theorem follows
immediately from the degree counting of the scattering equations.
\end{proof}

The upshot of the above discussion is that we can reduce any desired polynomial
with degree larger than $d^*=(n-3)(n-4)/2$ to a polynomial with degree equal or
less than $d^*$ towards an $H$-basis. The algorithm is recursive: given a
polynomial $f$ with $\deg f=d>d^*$, since $\In(f) \in \langle\In(I)\rangle $, 
\begin{equation}
  \label{eq:9}
  \In(f) =\sum_{i=1}^{n-3} q_i^{(d)} \In(f_i)\;,
\end{equation}
where $\deg q_i^{(d)}=d-\deg f_i$. Since the degrees are confined by virtue of
the $H$-basis, the $q_i^{(d)}$'s can be obtained by just solving a finite system
of linear equations. Define
\begin{equation}
  \label{eq:10}
  \tilde f =f- \sum_{i=1}^{n-3} q_i^{(d)} f_i\;.
\end{equation}
By construction, $\tilde f$ is polynomial with degree less than $d$, because we
subtracted off the leading part. Repeat at most $d-d^*$ times and collect the
intermediate coefficients, 
\begin{equation}
  \label{se_reduction}
  f=\sum_{i=1}^{n-3} q_i f_i + r\;,
\end{equation}
where $r$---the remainder after the division---is a polynomial with the degree
at most $(n-3)(n-4)/2$. With this procedure, $r$ is uniquely determined as a
consequence of the $H$-basis.

\section{Global Residues and the Bezoutian}
The $H$-basis property of the polynomial scattering equations allow us to prove
an exciting empirical observation \cite{Sogaard:2015dba}: the CHY formula
produces a global residue proportional to merely a single monomial coefficient,
and the constant of proportionality is universal. 

For the benefit of the reader, we recall the notion of a global residue with
emphasis on polynomials. We refer to the classical text books
\cite{MR507725,MR0463157} and related applications \cite{Sogaard:2015dba} (see
also
refs.~\cite{Sogaard:2013fpa,Sogaard:2014ila,Sogaard:2014oka,Johansson:2015ava}).
A individual (local) residue may require algebraic extensions such as
$\sqrt{2}$. On the contrary, a global residue is a manifestly rational quantity
of the monomial coefficients. Let $I = \avg{f_1,\dots,f_n}$ be a
zero-dimensional polynomial ideal, so $R/I$ is a finite-dimensional $\C$-linear
space. The global residue is a linear map $R/I\to\C$ that computes the total sum
of individual residues,
\begin{align}
\Res_{\{f_1,\dots,f_n\}}(H)\equiv
\sum_{\xi_i\in\mc Z(I)}\Res_{\{f_1,\dots,f_n\},\xi_i}(H)\;.
\label{EQ_GLOBAL_RESIDUE}
\end{align}
We can make connection with the CHY formalism and the scattering equations by
noting that
\begin{align}
H(z)\equiv 
\frac{z_2}{z_{n-1}}
\prod_{2\leq a<b\leq n-1}\!(z_a-z_b)
\prod_{a=2}^{n-2}\frac{z_a}{(z_a-z_{a+1})^2}
\times\tilde{\mc I}\;,
\end{align}
and therefore, eq.~\eqref{EQ_CHY_FORMULA} equals
$\Res{}_{\{h_1,\dots,h_{n-3}\}}(N)$. In particular, if $N$ is not a polynomial,
but a rational function $N=N_1/N_2$ and $N_2$ is nonvanishing on $\mc Z(I)$,
then
\cite{Sogaard:2015dba}
\begin{align}
\Res_{\{f_1,\dots,f_n\}}(H_1/H_2)\equiv
\Res_{\{f_1,\dots,f_n\}}(H_1G_2)\;.
\end{align}
Here, $G_2$ is the polynomial inverse of $N_2$ modulo $I$.

Efficient algebraic evaluation of global residues without computing the
individual residues is facilitated by the following theorem \cite{MR507725}.
\begin{thm}[Global Duality]
$\avg{\bullet,\bullet} : R/I\otimes R/I\to\C$ defined by
\begin{align}
\avg{N_1,N_2}\equiv\Res(N_1\cdot N_2)
\end{align}
is a nondegenerate inner product.
\end{thm}
Let $\{e_i\}$ be a basis for $R/I$. The strength of this theorem is that it
requires the existence of a {\it dual basis} $\{\Delta_i\}$ in $R/I$, defined by
the orthonormality conditions,
\begin{align}
\avg{e_i,\Delta_j} = \delta_{ij}\;.
\end{align}
Indeed, if we decompose $[N] = \sum_i\lambda_i e_i$ and
$1 = \sum_i\mu_i\Delta_i$ with $\lambda_i,\mu_i\in\C$, a tractable expression
emerges,
\begin{align}
\Res{}_{\{f_1,\dots,f_n\}}(N) = 
\sum_{i,j}\lambda_i\mu_j\avg{e_i,\Delta_j} = \sum_i \lambda_i\mu_i\;.
\label{EQ_GLOBAL_RES_BEZ}
\end{align}
The dual and canonical bases are obtained algorithmically by means of the
Gr{\"o}bner basis method and the Bezoutian matrix~\cite{Sogaard:2015dba}. This
method boils down the problem to taking linear combinations of monomial
coefficients of the numerator in question. Henceforth we will restrict attention
to a special case. If we randomly write down a zero-dimensional ideal, none of the
entries of the dual basis may be constant. But if $\Delta_1$ is a scalar, the
decomposition of unity over the dual basis becomes trivial and the global
residue \eqref{EQ_GLOBAL_RES_BEZ} truncates to a single term,
\begin{align}
\Res_{\{f_1,\dots,f_n\}}(N) = \lambda_1/\Delta_1\;.
\end{align}
We will momentarily show that an $H$-basis gives rise to this particular result.
Our starting point is the Euler-Jacobi vanishing theorem. 
\begin{thm}[Euler-Jacobi]
Suppose $I = \avg{f_1,\dots,f_n}$ is a zero-dimensional ideal whose generators
form an $H$-basis. Then, for any $N\in R$,
\begin{align}
\Res_{\{f_1,\dots,f_n\}}(N) = 0\;,
\end{align}
if $\deg(N) < d^*$, where $d^*\equiv\sum_{i=1}^n\deg f_i-n$ is the critical
degree.
\end{thm}

Accordingly, for the $n$-point tree-level polynomial scattering equations, if
the degree of the numerator is strictly less than $d^* = (n-3)(n-4)/2$, the
global residue vanishes identically. This observation leads us the following
theorem and corollary.
\begin{thm}
\label{THM_CANONICAL_BASIS_DEGREE}
Let $I = \avg{h_1,\dots,h_{n-3}}$ be the ideal generated by the polynomial
scattering equations. The canonical linear basis for $I/R$ must contain a
monomial of degree at least $d^*$.
\end{thm}
\begin{proof}
If not, then all monomials in the canonical linear basis have the degree
strictly less than $d^*$. By Euler-Jacobi's theorem, $\avg{1,m} = 0$, for every
monomial $m$ in the canonical basis. This is a contradiction of the
non-degenerate property of the inner product.
\end{proof}
\begin{corollary}
Let $I = \avg{h_1,\dots,h_{n-3}}$ be the ideal generated by the polynomial
scattering equations. Then, the dual basis of $R/I$ must contain a constant.
\label{COR_DUAL_BASIS}
\end{corollary}
\begin{proof}
The $k$th row of the Bezoutian matrix $B$ has degree $k$, ergo we have the bound
$\deg(\det B)\leq d^*$. The rest follows immediately from
Theorem~\ref{THM_CANONICAL_BASIS_DEGREE}.
\end{proof}

We have thus confirmed that the global residue with respect to an $H$-basis of
any polynomial $N\in R$ is always dictated entirely by the leading term of the
Gr{\"o}bner basis normal form of $N$. More specifically,
Corollary~\ref{COR_DUAL_BASIS} implies that
\begin{align}
\Res_{\{h_1,\dots,h_{n-3}\}}(N) = [N]_{z_{n-1}^{d^*}}/\Delta_1\;,
\end{align}
where the subscript indicates that only a single coefficient is
extracted. 

\section{Global Residues and $H$-basis}
Alternatively, besides the Bezoutian matrix and Gr\"obner basis approach, we can
also use the $H$-basis to calculate global residues with respect to the
scattering equations for polynomial numerators. Given $N$ as a degree $d$
polynomial, if $d>d^*=(n-3)(n-4)/2$, using the integrand reduction algorithm,
\begin{equation}
N=\sum_{i=1}^{n-3} q_i h_i +\tilde N\;.
\end{equation}
where $\tilde N$ is a polynomial with the degree at most $d^*=(n-3)(n-4)/2$.
From the preceding $H$-basis discussion, 
\begin{equation}
\Res_{\{h_1,\dots,h_{n-3}\}}(N) = \Res_{\{h_1,\dots,h_{n-3}\}}(
\In(\tilde N))\;,
\end{equation}
and $\In(N)$ consists monomials with the degree $d^*$. From the proper map
theorem \cite{CDS} of $H$-basis,
\begin{equation}
\Res_{\{h_1,\dots,h_{n-3}\}}(N) = 
\Res_{\{\In(h_1),\dots,\In (h_{n-3})\}}(\In(\tilde N))\;.
\end{equation}
Note that the $\In(h_i)$'s have only one common zero, namely at $(0,\dots,0)$.
Hence we just need to evaluate the residue at {\it one point}. Furthermore from
the $H$-basis graded algebra \cite{MR2161985}, 
\begin{equation}
  \label{eq:7}
  \dim_{\mathbb C} S_{d^*}-\dim_{\mathbb C} (S_{d^*}\cap\langle \In(I)\rangle)=1\;.
\end{equation}
Consequently, if a degree-$d^*$ monomial's residue is obtained and nonzero, all
other degree-$d^*$ monomial's residues are obtained from linear relations. Such
a residue can be found using the transformation law from algebraic geometry.
\begin{proposition}
For the $n$-point scattering equations in polynomial form, with the gauge
fixing $z_1 \to 1$, $z_{n-1} \to 0$ and $z_{n} \to \infty$,
\begin{equation}
  \label{}
  \Res_{\{h_1,\dots,h_{n-3}\}}(z_3 z_4^2 \ldots z_{n-2}^{n-4}) =
  \frac{(-1)^{(n-3)(n-4)/2}}{\prod_{j=2}^{n-2} \tilde s_{j,j+1,\ldots, n-2}}\;,
\end{equation}
where $\tilde s_{j,j+1,\ldots, n-2} \equiv (k_j+k_{j+1} +\ldots+k_{n-2}+k_n)^2$. 
\end{proposition}
\begin{proof}
Since the polynomial scattering equations form an $H$-basis,
$\In(h_i)\in \langle z_2,\dots,z_{n-2}\rangle$, $1\leq i\leq n-3$. That is,
\begin{equation}
  \In(h_i) =\sum_{j=2}^{n-2} a_{ij} z_j\;.
\end{equation}
Choosing the matrix $A=(a_{ij})$ to be upper triangular, the determinant becomes,
\begin{equation}
  \label{eq:11}
  \det A=(-1)^\frac{(n-3)(n-4)}{2} z_3 z_4^2 \ldots z_{n-2}^{n-4}
  \prod_{j=2}^{n-2} \tilde s_{j,j+1,\ldots,n-2}\;.
\end{equation}
Hence, the result follows from transformation law \cite{MR507725}.
\end{proof}

Using this straightforward approach, we are able to get the residue of any
polynomial numerator in analytic form using the $H$-basis.

\section{Conclusion}
We have uncovered and proved that the polynomial form of the scattering
equations is an $H$-basis. We have explored and emphasized several compelling
implications of this observation, and briefly compared with the presumably more
familiar Gr{\"o}bner basis, which can be computationally expensive to obtain.

In particular, the $H$-basis enables us to perform reductions of high-degree
multivariate polynomials without the need for a Gr{\"o}bner basis. More
concretely, in connection with the scattering equations we have shown that any
monomial with degree greater than $d^* = (n-3)(n-4)/2$ can always be reduced to
a polynomial of degree at most $d^*$, modulo the $H$-basis. This procedure only
involves linear algebra. The $H$-basis greatly enhances our ability to compute
global residues and thus calculate scattering amplitudes in the CHY framework.

In this direction we have also proved a conjecture recently made in
ref.~\cite{Sogaard:2015dba}, namely that the dual basis associated with the
polynomial scattering equations always contains a constant and that hence any
global residue is just one rational monomial coefficient, multiplied by a
universal factor. In a forthcoming paper \cite{BSZ} we expect to tabulate
analytic expressions for many of the CHY global residues, also at loop level.

It remains intriguing to gain a complete insight into the algebraic geometry
underlying the whole CHY formalism. We anticipate that the explicit
identification of the polynomial scattering equations as an $H$-basis paves the
way for new exciting advances in this direction.

{\bf Acknowledgments:}
We acknowledge useful discussions with N. Beisert, L. Dolan, P. Goddard, A.
Georgoudis and K. Larsen. M.S. thanks ETH Z{\"u}rich, Kavli Institute for
Theoretical Physics, and Mainz Institute for Theoretical Physics for
hospitality. The work of M.S. is supported by the Danish Council for Independent
Research under contract No. DFF-4181-00563. Y.Z. is supported by an Ambizione
grant (PZ00P2\_161341) from the Swiss National Foundation and in part by the
National Science Foundation under Grant No. NSF PHY11-25915. Y.Z. acknowledges
the KITP workshop ``LHC Run II and the Precision Frontier''.

\end{document}